\documentclass[envcountsect,envcountsame,runningheads]{llncs}

\usepackage[utf8]{inputenc}
\usepackage[english,british]{babel}

\usepackage{amsmath,amssymb,braket,enumerate}
\usepackage{wasysym,stmaryrd}
\usepackage{hyperref}
\usepackage{ifthen}
\usepackage{tikz,wrapfig}

\newcommand{\dom}{\operatorname{dom}}

\newcommand{\ar}{\operatorname{ar}}
\newcommand{\Nat}{\mathbb{N}}

\renewcommand{\epsilon}{\varepsilon}
\newcommand{\restrict}{\mathord{\restriction}}
\newcommand{\aut}[1]{\mathcal{#1}}

\newcommand{\Rel}{\mathcal{R}}
\newcommand{\struct}[1]{\mathfrak{#1}}
\newcommand{\FO}[1]{\mathsf{FO}\ifthenelse{\equal{#1}{}}{}{_{#1}}}
\newcommand{\MSO}[1]{\mathsf{MSO}\ifthenelse{\equal{#1}{}}{}{_{#1}}}
\newcommand{\Int}{\mathcal{I}}
\newcommand{\signW}[1]{\mathsf{W}#1}
\newcommand{\signT}[1]{\mathsf{T}#1}

\newcommand{\Pre}{\operatorname{Pre}}

\newcommand{\trees}[2][\empty]{T_{#2\ifthenelse{\equal{#1}{\empty}}{}{,#1}}}
\newcommand{\prefix}{\prec}
\newcommand{\prefixeq}{\preceq}

\renewcommand{\phi}{\varphi}

\newcommand{\cmp}{\operatorname{cmp}\nolimits}

\newcommand{\killenumspace}{\vspace{-0.5ex}}

\newcommand{\run}[3]{#1\in [#2,#3]}

\newcounter{mysection}

\newtheorem{lemma*}{Lemma}[mysection]

\begin{document}

\title{Word Automaticity of Tree Automatic Scattered Linear Orderings Is Decidable}
\author{Martin Huschenbett}
\institute{Institut f\"ur Theoretische Informatik, Technische Universit\"at Ilmenau, Germany\\\email{martin.huschenbett@tu-ilmenau.de}}
\maketitle

\begin{abstract}
\noindent A tree automatic structure is a structure whose domain can be encoded by a regular tree language such that each relation is recognisable by a finite automaton processing tuples of trees synchronously. Words can be regarded as specific simple trees and a structure is word automatic if it is encodable using only these trees. The question naturally arises whether a given tree automatic structure is already word automatic. We prove that this problem is decidable for tree automatic scattered linear orderings. Moreover, we show that in case of a positive answer a word automatic presentation is computable from the tree automatic presentation.
\end{abstract}


\section{Introduction}

%

The fundamental idea of automatic structures can be traced back to the 1960s when B\"uchi, Elgot, Rabin, and others used finite automata to provide decision procedures for the first-order theory of Presburger arithmetic $(\Nat;+)$ and several other logical problems. Hodgson generalised this idea to the concept of \emph{automaton decidable} first-order theories. Independently of Hodgson and inspired by the successful employment of finite automata and their methods in group theory, Khoussainov and Nerode \cite{KN95} initiated the systematic investigation of \emph{automatic structures}. Recalling the efforts from the 1960s, Blumensath \cite{Blu99} extended this concept notion beyond finite automata to finite automaton models recognising infinite words, finite trees, or infinite trees.

Basically, a countable relational structure is \emph{tree automatic} or \emph{tree automatically presentable} if its elements can be encoded by finite trees in such a way that its domain and its relations are recognisable by finite automata processing either single trees or tuples of trees synchronously. A structure is \emph{word automatic} if its elements can be encoded using only specific simple trees which effectively represent words. In contrast to the more general concept of \emph{computable structures} and based on the strong closure properties of recognisability, automatic structures provide pleasant algorithmic features. In particular, they possess decidable first-order theories.

Due to this latter fact, the concept of automatic structures gained a lot attention which led to noticeable progress (cf.~\cite{BGR11,Rub08}). Automatic presentations were found for many structures, some structures where shown to be tree but not word automatic, for instance Skolem arithmetic $(\Nat;\times)$, whereas other structures, like the random graph, were proven to be neither word nor tree automatic. For some classes of structures it was even possible to characterise its automatic members, for example an ordinal is word automatic respectively tree automatic precisely if it is less than $\omega^\omega$ respectively $\omega^{\omega^\omega}$. Certain extensions of first-order logic were shown to preserve decidability of the corresponding theory. The question whether two automatic structures are isomorphic turned out to be highly undecidable in general as well as for some restricted classes of structures. At the same time, the isomorphism problem for word automatic ordinals was proven to be decidable. Last but not least, the different classes of automatic structures was characterised by means of interpretations in universal structures.

Due to the fact that word automaticity is a special case of tree automaticity, the question naturally arises whether a given tree automatic structure is already word automatic. As far as we know, this problem was neither solved in general nor for any restricted class of structures. For that reason, we investigate the respective question for scattered linear orderings in this paper. Actually, we prove the corresponding problem to be decidable and our main result is as follows:

\begin{theorem}
\label{thm:main}
Given a tree automatic presentation $\aut P$ of a scattered linear ordering $\struct L$, it is decidable whether $\struct L$ is word automatic. In case $\struct L$ is word automatic, one can compute a word automatic presentation of $\struct L$ from $\aut P$.
\end{theorem}

\noindent Since every well-ordering is scattered, this result still holds if $\struct L$ is assumed to be an ordinal. The proof of Theorem~\ref{thm:main} splits into three parts. First, we introduce the notion of \emph{slim} tree languages and prove this property to be decidable (Theorem~\ref{thm:slim_decidable}). Second, we show that a slim domain is sufficient for a tree automatic structure to be word automatic (Theorem~\ref{thm:slim_TA_are_WA}). Last, we demonstrate that this condition is also necessary in case of scattered linear orderings (Theorem~\ref{thm:fat_TA_ord_are_not_WA}). Altogether, Theorem~\ref{thm:main} follows from the three mentioned theorems.\footnote{Proofs of all seemingly unproven lemmas as well as the interpretations from Section~\ref{sec:slim_TA_are_WA} can be found in the appendix.}

\section{Background}

In this section we recall the necessary notions of logic, automatic structures (cf.~\cite{BGR11,Rub08}), tree automata (cf.~\cite{GS97}), and linear orderings. We agree that the natural numbers $\Nat$ include $0$ and that $[m,n]=\{m,m+1,\dotsc,n\}\subseteq\Nat$ for all $m,n\in\Nat$.

\paragraph{Logic.}

A \emph{(relational) signature} $\tau=(\Rel,\ar)$ is a finite set $\Rel$ of \emph{relation symbols} together with a map $\ar\colon\Rel\to\Nat$ assigning to each $R\in\Rel$ its \emph{arity} $\ar(R)\geq 1$. A $\tau$-structure $\struct A=\bigl(A;(R^{\struct A})_{R\in\Rel}\bigr)$ consists of a set $A=\dom(\struct A)$, its \emph{domain}, and an $\ar(R)$-ary relation $R^{\struct A}\subseteq A^{\ar(R)}$ for each $R\in\Rel$. \emph{First order logic} $\FO\tau$ over $\tau$ is defined as usual, including an equality predicate. A \emph{sentence} is a formula without free variables. Writing $\phi(\bar x)$ means that all free variables of the formula $\phi$ are among the entries of the tuple $\bar x=(x_1,\dotsc,x_n)$. The set $\phi^{\struct A}$ is comprised of all $\bar a\in A^n$ satisfying $\struct A\models\phi(\bar a)$, where the latter is defined as usual.

\paragraph{Automatic Structures.}

The set of all \emph{(finite) words} over an alphabet $\Sigma$ is $\Sigma^\star$, the \emph{empty word} is $\epsilon$, and the \emph{length} of $w$ is $|w|$. Subsets of $\Sigma^\star$ are called \emph{languages} and $L\subseteq\Sigma^\star$ is \emph{regular} if it can be \emph{recognised} by some (non-deterministic) finite automaton.

Let $\Box\not\in\Sigma$ be a new symbol and $\Sigma_\Box=\Sigma\cup\{\Box\}$. For $n\geq 1$ consider an $n$-tuple $\bar w=(w_1,\dotsc,w_n)\in(\Sigma^\star)^n$ of words with $w_i=a_{i,1}a_{i,2}\dotsc a_{i,m_i}$ for all $\run i1n$. Let $m=\max\{m_1,\dotsc,m_n\}$ and $a_{i,j}=\Box$ for $\run j{m_i+1}m$. The \emph{convolution} of $\bar w$ is the word $\otimes\bar w=\bar a_1\dotsc\bar a_m\in(\Sigma_\Box^n)^\star$ with $\bar a_j=(a_{1,j},\dotsc,a_{n,j})\in\Sigma_\Box^n$ for all $\run j1m$. An $n$-ary relation $R\subseteq(\Sigma^\star)^n$ is \emph{automatic} if the language $\otimes R\subseteq(\Sigma_\Box^n)^\star$, which is comprised of all $\otimes\bar w$ with $\bar w\in R$, is regular.

A $\tau$-structure $\struct A$ with $\dom(\struct A)\subseteq\Sigma^\star$ is \emph{(word) automatic} if $\dom(\struct A)$ is regular and $R^{\struct A}$ is automatic for all $R\in\Rel$. A \emph{(word) automatic presentation} of $\struct A$ is a tuple $\bigl(\aut A_{\dom};(\aut A_R)_{R\in\Rel}\bigr)$ of finite automata such that $\aut A_{\dom}$ recognises $\dom(\struct A)$ and $\aut A_R$ recognises $\otimes R^{\struct A}$. Abusing notation, we call any structure $\struct B$ which is isomorphic to some word automatic structure $\struct A$ also \emph{(word) automatic}.

\paragraph{Tree Automata.}

A \emph{tree domain} is a non-empty, finite, and prefix-closed subset $D\subseteq\{0,1\}^\star$ satisfying $u0\in D$ iff $u1\in D$ for all $u\in D$. A \emph{tree} over $\Sigma$ is a map $t\colon D\to\Sigma$ where $\dom(t)=D$ is a tree domain. The set of all trees is denoted by $\trees\Sigma$ and its subsets are called \emph{(tree) languages}. For some $t\in\trees\Sigma$ and $u\in\dom(t)$ the \emph{subtree} of $t$ \emph{rooted} at $u$ is the tree $t\restrict u\in\trees\Sigma$ defined by
\begin{equation*}
	\dom(t\restrict u) = \Set{ v\in\{0,1\}^\star | uv\in\dom(t) }
	\quad\text{and}\quad
	(t\restrict u)(v) = t(uv)\,.
\end{equation*}

\noindent A \emph{(deterministic bottom-up) tree automaton} $\aut A=(Q,\iota,\delta,F)$ over $\Sigma$ consists of a finite set $Q$ of \emph{states}, a \emph{start state function} $\iota\colon\Sigma\to Q$, a \emph{transition function} $\delta\colon \Sigma\times Q\times Q\to Q$, and a set $F\subseteq Q$ of \emph{accepting states}. For each $t\in\trees\Sigma$ a state $\aut A(t)\in Q$ is defined recursively by $\aut A(t)=\iota\bigl(t(\epsilon)\bigr)$ if $\dom(t)=\{\epsilon\}$ and $\aut A(t)=\delta\bigl(t(\epsilon),\aut A(t\restrict 0),\aut A(t\restrict 1)\bigr)$ otherwise.
The language \emph{recognised} by $\aut A$ is the set of all $t\in\trees\Sigma$ with $\aut A(t)\in F$. A language $L\subseteq\trees\Sigma$ is \emph{regular} if it can be recognised by some tree automaton.

The \emph{convolution} of $\bar t=(t_1,\dotsc,t_n)\in(\trees\Sigma)^n$ is the tree $\otimes\bar t\in\trees{\Sigma_\Box^n}$ defined by $\dom(\otimes\bar t)=\dom(t_1)\cup\dotsb\cup\dom(t_n)$ and $(\otimes\bar t)(u)=\bigl(t'_1(u),\dotsc,t'_n(u)\bigr)$, where $t'_i(u)=t_i(u)$ if $u\in\dom(t_i)$ and $t'_i(u)=\Box$ otherwise. A relation $R\subseteq(\trees\Sigma)^n$ is \emph{automatic} if the language $\otimes R\subseteq\trees{\Sigma_\Box^n}$ is regular.

\emph{Tree automatic structures} and \emph{tree automatic presentations} are defined like in the word automatic case, but based on trees and tree automata.

\paragraph{Linear Orderings.} A \emph{linear ordering} is a structure $\struct A=\bigl(A;<^{\struct A}\bigr)$ where $<^{\struct A}$ is a \emph{strict} linear order relation on $A$. The ordering $\struct A$ is \emph{scattered} if $(\mathbb{Q};<)$ cannot be embedded into $\struct A$. Obviously, every well-ordering is scattered. For any two linear orderings $\struct A$ and $\struct B$ we define another linear ordering $\struct A\cdot\struct B$ by $\dom(\struct A\cdot\struct B)=\dom(\struct A)\times\dom(\struct B)$ and $(a_1,b_1) <^{\struct A\cdot\struct B} (a_2,b_2)$ iff either $a_1 <^{\struct A} a_2$ or $a_1=a_2$ and $b_1 <^{\struct B} b_2$. Finally, if $\struct A_1$ can be embedded into $\struct B_1$ and $\struct A_2$ into $\struct B_2$, then $\struct A_1\cdot\struct A_2$ can be embedded into $\struct B_1\cdot\struct B_2$.

\section{Slim and Fat Tree Languages}
\label{sec:slim_and_fat}

In this section, we introduce the notion of slim tree languages and show that it is decidable whether the language recognised by a given tree automaton is slim.

\begin{definition}
The \emph{thickness} $\diameter(t)$ of a tree $t\in T_\Sigma$ is the maximal number of nodes on any level, i.e.,
\begin{equation*}
	\diameter(t) = \max \Set{ \bigl\vert\dom(t)\cap\{0,1\}^\ell\bigr\vert | \ell\geq 0 } \in\Nat\,.
\end{equation*}
For every $K\geq 1$ the set of all $t\in\trees\Sigma$ with $\diameter(t)\leq K$ is denoted by $\trees[K]\Sigma$. A~tree language $L\subseteq\trees\Sigma$ is \emph{slim} if there exists some $K\geq 1$ such that $L\subseteq\trees[K]\Sigma$, otherwise $L$ is \emph{fat}.
\end{definition}

\noindent A tree automaton $\aut A$ is \emph{reduced} if for every state $q$ of $\aut A$ there is a tree $t\in\trees\Sigma$ with $\aut A(t)=q$. For every tree automaton $\aut A$ one can compute a reduced tree automaton which recognises the same language and has no more states than $\aut A$.

\begin{theorem}
\label{thm:slim_decidable}
Given a reduced tree automaton $\aut A$, it is decidable whether the tree language $L$ recognised by $\aut A$ is slim or fat. If $L$ is slim, then $L\subseteq\trees[2^{n-1}]\Sigma$, where $n$ is the number of states of $\aut A$.
\end{theorem}

\noindent For the rest of this section we fix a reduced tree automaton $\aut A=(Q,\iota,\delta,F)$. The proof of Theorem~\ref{thm:slim_decidable} essentially depends on an inspection of the directed graph $G_{\aut A}=(Q,E_{\aut A})$ with
\begin{equation}
\label{eq:def_E_A}
	(p,q)\in E_{\aut A} \quad\text{iff}\quad \exists a\in\Sigma,r\in Q\colon
			\delta(a,p,r)=q \text{ or } \delta(a,r,p)=q\,.
\end{equation}
Clearly, this graph is computable from $\aut A$. The lemma below is shown by applying the idea of pumping to tree automata. Therein, the \emph{height} $h(t)$ of a tree $t\in\trees\Sigma$ is the number
\begin{equation*}
	h(t) = \max \Set{ \vert u\vert | u\in\dom(t) } \in \Nat\,.
\end{equation*}

\begin{lemma}
\label{lemma:inflang_characterization}
For every $q\in Q$ the following are equivalent:
\begin{enumerate}[(1)]
\item there are infinitely many $t\in\trees\Sigma$ satisfying $\aut A(t)=q$,
\item there is a tree $t\in\trees\Sigma$ satisfying $h(t)\geq n$ and $\aut A(t)=q$, where $n=|Q|$,
\item $G_{\aut A}$ contains a cycle from which $q$ is reachable.
\end{enumerate}
\end{lemma}

\noindent An edge $(p,q)\in E_{\aut A}$ is \emph{special} if in the definition of $E_{\aut A}$ in Eq.~\eqref{eq:def_E_A} the state $r\in Q$ can be chosen such that it satisfies the conditions of Lemma~\ref{lemma:inflang_characterization} (for $r$ in place of $q$). Since condition (3) is decidable, it is decidable whether an edge is special. The key idea for proving Theorem~\ref{thm:slim_decidable} is stated by the following lemma:

\begin{lemma}
\label{lemma:fatlang_characterization}
The following are equivalent:
\begin{enumerate}[(1)]
\item the tree language $L$ recognised by $\aut A$ is fat,
\item there is a tree $t\in L$ satisfying $\diameter(t)> 2^{n-1}$, where $n=|Q|$,
\item $G_{\aut A}$ contains a cycle including a special edge and from which $F$ is reachable.
\end{enumerate}
\end{lemma}

\noindent The proof of this lemma works similar to the one of Lemma~\ref{lemma:inflang_characterization}. Since condition~(3) is decidable given $\aut A$ as input, Theorem~\ref{thm:slim_decidable} follows.

\section{Slim Tree Automatic Structures Are Word Automatic}
\label{sec:slim_TA_are_WA}

This section is devoted to the proof of the following theorem:

\begin{theorem}
\label{thm:slim_TA_are_WA}
Let $\struct A$ be a tree automatic structure such that $\dom(\struct A)$ is slim. Then, $\struct A$ is already word automatic and one can compute a word automatic presentation of $\struct A$ from a tree automatic presentation of $\struct A$.
\end{theorem}

\noindent The idea of the proof is the following. Let $K\geq 1$ be such that ${\dom(\struct A)\subseteq\trees[K]\Sigma}$. We give an alphabet $\widehat\Sigma$ and a one-to-one map ${C\colon\trees[K]\Sigma\to\widehat\Sigma^\star}$, the \emph{encoding}, such that $C(L)$ is regular for all regular $L\subseteq\trees[K]\Sigma$ (Proposition~\ref{prop:C_preserves_regularity}) and $C(R)$ is automatic for all automatic relations $R\subseteq(\trees[K]\Sigma)^n$ (Proposition~\ref{prop:C_preserves_automaticity}). Thus, the structure $C(\struct A)$ is word automatic. A word automatic presentation of $C(\struct A)$ is computable since both propositions are effective and Theorem~\ref{thm:slim_decidable} allows for computing a suitable~$K$. Although it is possible to show both propositions using automata, it is much more convenient to accomplish this by means of logic.

\subsection{Monadic Second Order Logic}

\emph{Monadic second order logic} $\MSO\tau$ extends $\FO{\tau}$ by \emph{set variables}, which range over subsets of the domain and are denoted by capital letters, quantifiers for these variables, and the formula ``$x\in X$'' (cf.~\cite{Tho97}). Let $\tau=(\Rel,\ar)$ and $\tau'$ be two signatures. An \emph{($\MSO{}$-)interpretation} of a $\tau$\nobreakdash-structure~$\struct A$ in a $\tau'$\nobreakdash-structure~$\struct B$ is a pair $\langle f,\Int\rangle$ comprised of a one-to-one map ${f\colon\dom(\struct A)\to\dom(\struct B)}$ and a tuple $\Int=\bigl(\Delta;(\Phi_R)_{R\in\Rel}\bigr)$ of $\MSO{\tau'}$-formulae with free $\FO{}$-variables only such that $f\bigl(\dom(\struct A)\bigr)=\Delta^{\struct B}$ and $f\bigl(R^{\struct A}\bigr)=\Phi_R^{\struct B}$ for each $R\in\Rel$. In fact, $f$ induces an isomorphism between $\struct A$ and $\Int(\struct B)=\bigl(\Delta^{\struct B};(\Phi_R^{\struct B})_{R\in\Rel}\bigr)$. Replacing in an $\MSO\tau$-formula $\phi(\bar x)$ all symbols $R\in\Rel$ with $\Phi_R$ and relativising quantifiers to $\Delta$ yields an $\MSO{\tau'}$-formula $\phi^\Int(\bar x)$ satisfying $\struct A\models\phi(\bar a)$ iff $\struct B\models\phi^\Int\bigr(f(\bar a)\bigr)$ for all $\bar a\in A^n$.

For an alphabet $\Sigma$ the signature $\signW\Sigma$ consists  of one binary relation symbol~$\leq$ and a unary symbol $P_a$ for each $a\in\Sigma$. Every word $w=a_1a_2\dotsc a_{|w|}\in\Sigma^\star$ is regarded as a $\signW\Sigma$-structure with domain $\dom(w)=\{1,\dotsc,|w|\}$, $\leq^w$ being the natural order on $\dom(w)$, and $i\in P_a^w$ iff $a_i=a$. For fixed numbers $m,r\in\Nat$, relations like $x=y+m$ and $x\equiv r\,(\bmod\, m)$ are expressible in $\MSO{\signW\Sigma}$. The language \emph{defined} by an $\MSO{\signW\Sigma}$-sentence $\Phi$ is the set of all $w\in\Sigma^\star$ with $w\models\Phi$.

The signature $\signT\Sigma$ is similar to $\signW\Sigma$ but contains two binary symbols $S_0$ and $S_1$ instead of $\leq$. Each tree $t\in\trees\Sigma$ is considered as a $\signT\Sigma$-structure with domain $\dom(t)$, $(u,v)\in S_d^t$ iff $ud=v$ ($d=0,1$), and $u\in P_a^t$ iff $t(u)=a$. The language \emph{defined} by some $\MSO{\signT\Sigma}$-sentence $\Phi$ is the set of all $t\in\trees\Sigma$ with $t\models\Phi$.

The following theorem holds for word languages as well as for tree languages:

\begin{theorem}[cf.~\cite{Tho97}]
\label{thm:regular_MSO}
A language $L$ is regular iff it is definable in $\MSO{}$, and both conversions, from automata to formulae and vice versa, are effective.
\end{theorem}

\pagebreak

\subsection{The Encoding and Preservation of Regularity}

\newcommand{\extree}{\begin{tikzpicture}[level distance=5mm,inner sep=0.2ex,every node/.style={minimum size=0}
	,level 1/.style={sibling distance=10mm}
	,level 2/.style={sibling distance=5mm}
	,level 3/.style={sibling distance=3mm}
	]
	\node {\texttt{a}}
		child {node {\texttt{b}}
			child {node {\texttt{c}}}
			child {node {\texttt{b}}
				child {node {\texttt{a}}}
				child {node {\texttt{c}}}
			}
		}
		child {node {\texttt{c}}
			child {node {\texttt{b}}}
			child {node {\texttt{a}}}
		};
\end{tikzpicture}}

\setlength{\intextsep}{0ex}
\setlength{\columnsep}{1.5em}

\begin{wrapfigure}{r}{2.1cm}
\centering
\hspace{2mm}
\begin{tikzpicture}[level distance=5mm,inner sep=0.2ex,every node/.style={minimum size=0}
	,level 1/.style={sibling distance=10mm}
	,level 2/.style={sibling distance=5mm}
	,level 3/.style={sibling distance=3mm}
	]
	\node {\texttt{a}}
		child {node {\texttt{b}}
			child {node {\texttt{c}}}
			child {node {\texttt{b}}
				child {node {\texttt{a}}}
				child {node {\texttt{c}}}
			}
		}
		child {node {\texttt{c}}
			child {node {\texttt{b}}}
			child {node {\texttt{a}}}
		};
\end{tikzpicture}
\vspace{-4ex}
\caption{The tree $t_{\mathrm{ex}}$.}
\label{fig:extree}
\end{wrapfigure}
\emph{For the rest of this section fix the $K\geq 1$ from above.} The first objective is to give the encoding $C\colon\trees[K]\Sigma\to\widehat\Sigma^\star$, where $\$$ is a new symbol and ${\widehat\Sigma=\Sigma\times\{0,1\}\cup\{\$\}}$. For a tree $t\in\trees[K]\Sigma$ of height $m=h(t)$ its \emph{encoding} $C(t)=\sigma_0\sigma_1\dotso\sigma_m$ is made up of $m+1$ blocks $\sigma_0,\dotsc,\sigma_m\in\widehat\Sigma^K$ describing the individual levels of $t$. More specifically, $\sigma_\ell$ consists of the labels of the $\ell$-th level from left to right, each enriched by a bit stating whether the corresponding node possesses children, and is padded up to length $K$ by $\$$ symbols. For example, the tree $t_{\mathrm{ex}}\in\trees{\{\mathtt{a},\mathtt{b},\mathtt{c}\}}$ in Figure~\ref{fig:extree} on the right satisfies $\diameter(t_{\mathrm{ex}})=4$ and is, under the assumption $K=5$, encoded by the word
\begin{equation*}
	C(t_{\mathrm{ex}}) = \langle\mathtt{a},\!1\rangle\$\$\$\$\,\langle\mathtt{b},\!1\rangle\langle\mathtt{c},\!1\rangle\$\$\$\,\langle\mathtt{c},\!0\rangle\langle\mathtt{b},\!1\rangle\langle\mathtt{b},\!0\rangle\langle\mathtt{a},\!0\rangle\$\,\langle\mathtt{a},\!0\rangle\langle\mathtt{c},\!0\rangle\$\$\$\,.
\end{equation*}

\noindent Formally, for each $\run\ell0m$ let $u_{\ell,1},\dotsc,u_{\ell,s_\ell}$ be the lexicographic enumeration (w.r.t. ${0<1}$) of $\dom(t)\cap\{0,1\}^\ell$. For $\run r1{s_\ell}$ we let $c_{\ell,r}=1$ if $u_{\ell,r}$ is an inner node, i.e. $u_{\ell,r}\{0,1\}\subseteq\dom(t)$, and $c_{\ell,r}=0$ if $u_{\ell,r}$ is a leaf. Finally, we put
\begin{equation*}
	\sigma_\ell=\langle t(u_{\ell,1}),c_{\ell,1}\rangle\langle t(u_{\ell,2}),c_{\ell,2}\rangle\dotso\langle t(u_{\ell,s_\ell}),c_{\ell,s_\ell}\rangle\$^{K-s_\ell}\,.
\end{equation*}
The main tool for studying the map $C\colon\trees[K]\Sigma\to\widehat\Sigma^\star$ is the following lemma:

\begin{lemma}
\label{lemma:interpretation_C}
For all $t\in\trees[K]\Sigma$ there is an $\MSO{}$-interpretation $\langle f_C,\Int_C\rangle$ of $t$ in~$C(t)$ such that $\Int_C$ does not depend on $t$.
\end{lemma}

\begin{proof}
Observe that for each inner node $u$ of $t$ the children of $u$ are the $(2s-1)$\nobreakdash-th and $2s$-th node on the next level, where $s$ is the number of inner nodes from left up to $u$ on its level. Formally, for an inner node $u_{\ell,r}$ we have $u_{\ell,r}d=u_{\ell+1,2s-1+d}$, where $d\in\{0,1\}$ and $s=c_{\ell,1}+\dotsb+c_{\ell,r}$. Based on this observation, one can give an interpretation $\langle f_C,\Int_C\rangle$ of $t$ in $C(t)$ such that $f_C(u_{\ell,r})=\ell\cdot K+r$.\qed
\end{proof}

\noindent As a first consequence, we obtain $t\cong\Int_C(t)=\Int_C(t')\cong t'$, and hence $t=t'$, for all $t,t'\in\trees[K]\Sigma$ with $C(t)=C(t')$. Thus, the encoding $C$ is one-to-one. The proof of Proposition~\ref{prop:C_preserves_regularity} is mainly based on Lemma~\ref{lemma:interpretation_C} and Lemma~\ref{lemma:characterisation_image_C} below.

\begin{proposition}
\label{prop:C_preserves_regularity}
Let $L\subseteq\trees[K]\Sigma$ be a regular language. Then, the language $C(L)\subseteq\widehat\Sigma^\star$ is also regular and one can compute a finite automaton recognising $C(L)$ from a tree automaton recognising $L$.
\end{proposition}

\begin{lemma}
\label{lemma:characterisation_image_C}
Let $\sigma\in\widehat\Sigma^\star$. There exists a tree $t\in\trees[K]\Sigma$ with $C(t)=\sigma$ iff $\sigma=\sigma_0\sigma_1\dotso\sigma_n$ for some $n\geq 0$ and $\sigma_0,\dotsc,\sigma_n\in\widehat\Sigma^K$ satisfying (a) and (b):
\begin{enumerate}[(a)]
\item $\sigma_\ell=\alpha_{\ell,1}\dotso\alpha_{\ell,s_\ell}\$^{K-s_\ell}$ for some $s_\ell\geq 1$ and $\alpha_{\ell,1},\dotsc,\alpha_{\ell,s_\ell}\in\Sigma\times\{0,1\}$ and for each $\run \ell0n$,
\item $s_0=1$, $s_{\ell+1}=2\cdot(c_{\ell,1}+\dotsb+c_{\ell,s_\ell})$ for $0\leq\ell<n$, and $c_{m,1}+\dotsb+c_{m,s_m}=0$, where $\alpha_{\ell,r}=\langle a_{\ell,r},c_{\ell,r}\rangle$.
\end{enumerate}
\end{lemma}

\begin{proof}
To see that $C(t)$ has the required shape, notice that (b) mainly reflects the relationship between the numbers of nodes on two adjacent levels. Conversely, if $\sigma\in\widehat\Sigma^\star$ is of the required shape, then there is a tree $t\in\trees\Sigma$ with $t\cong\Int_C(\sigma)$ and it turns out that $\diameter(t)\leq K$ and $C(t)=\sigma$.\footnote{More details on this can be found in Appendix~\ref{subsec:characterisation_image_C}.}\qed
\end{proof}

\begin{proof}[of Proposition~\ref{prop:C_preserves_regularity}]
Let $\Gamma_C$ be an $\MSO{\signW{\widehat\Sigma}}$-sentence which expresses the requirement on the shape of $\sigma$ from Lemma~\ref{lemma:characterisation_image_C}.
By Theorem~\ref{thm:regular_MSO}, there is an $\MSO{\signT\Sigma}$-sentence $\Phi$ defining $L\subseteq\trees[K]\Sigma$. Then, the $\MSO{\signW{\widehat\Sigma}}$-sentence $\Gamma_C\land\Phi^{\Int_C}$ defines $C(L)$ and, again by Theorem~\ref{thm:regular_MSO}, this language is regular. Finally, all employed constructions are effective.\qed
\end{proof}

\subsection{Preservation of Automaticity}

The purpose of this subsection is to complete the proof of Theorem~\ref{thm:slim_TA_are_WA}.

\begin{proposition}
\label{prop:C_preserves_automaticity}
Let $R\subseteq(\trees[K]\Sigma)^n$ be an automatic relation. Then, the relation $C(R)\subseteq(\widehat\Sigma^\star)^n$ is also automatic and one can compute a finite automaton recognising $\otimes C(R)$ from a tree automaton recognising $\otimes R$.
\end{proposition}

\noindent Basically, the key idea behind the proof is the same as for Proposition~\ref{prop:C_preserves_regularity} though it is more involved. Let $\bar t=(t_1,\dotsc,t_n)\in(\trees[K]\Sigma)^n$. Due to cardinality reasons, $\otimes\bar t$ is commonly not directly interpretable in $\otimes C(\bar t)$ but only in an $n$-fold copy of $\otimes C(\bar t)$. This is formalised by means of the one-to-one monoid morphism
\begin{equation*}
	H\colon(\widehat\Sigma_\Box^n)^\star\to(\widehat\Sigma_\Box^n)^\star,\bar\alpha_1\dotso\bar\alpha_m\mapsto\bar\alpha_1^n\dotso\bar\alpha_m^n\,.
\end{equation*}
The interpretation of $\otimes\bar t$ in $H\bigl(\otimes C(\bar t)\bigr)$ embraces two aspects which are better considered separately. Thus, we define an intermediate structure $\amalg\bar t$ which extends the disjoint union of the $t_i$'s on domain $\dom(\amalg\bar t)=\bigcup_{\run i1n} \{i\}\times\dom(t_i)$ by a binary relation $L^{\amalg\bar t}$, relating all $(i,u)$ and $(j,v)$ with $|u|=|v|$, and unary relations $Q_i^{\amalg\bar t}=\{i\}\times\dom(t_i)$ for each $\run i1n$. Altogether, we give several interpretations whose formulae naturally do not depend on the specific choice of~$\bar t$. An overview of the whole setting is depicted in Figure~\ref{fig:intepretations}.

\begin{figure}
\centering
\vspace{3ex}
\begin{tikzpicture}[->,semithick,inner sep=0.5mm]
	\node (conv) at (0cm,0cm) {$\otimes(t_1,\dotsc,t_n)$};
	\node (sum) at (4cm,0cm) {$\amalg(t_1,\dotsc,t_i,\dotsc,t_n)$};
	\node (sum') at (3.7cm,0cm) {\phantom{$\amalg(t_1,\dotsc,t_i,\dotsc,t_n)$}};
	\node (H) at (9.5cm,0cm) {$H\bigl(\otimes\bigl(C(t_1),\dotsc,C(t_n)\bigr)\bigr)$};
	\node (C) at (6.6cm,1cm) {$C(t_i)$};
	
	\draw (conv) edge node[above] {$\langle f_\amalg,\Int_\amalg\rangle$} (sum)
	      (sum) edge node[above] {$\langle f_H,\Int_H\rangle$} (H)
	      (sum') edge[dashed] node[above,rotate=19] {$\langle f_{C,i},\Int_C\rangle$} (C)
	      (C) edge node[above,rotate=-19] {$\langle f_{\otimes,i},\Int_{\otimes,i}\rangle$} (H);
\end{tikzpicture}
\caption{Interpretations involved in proving Proposition~\ref{prop:C_preserves_automaticity}.}
\label{fig:intepretations}
\end{figure}
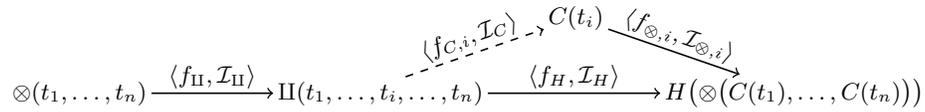

\paragraph{The Interpretation $\langle f_\amalg,\Int_\amalg\rangle$.} The main idea is to construct an $\MSO{}$-formula $E(x,y)$ with $\amalg\bar t\models E\bigl((i,u),(j,v)\bigr)$ iff $u=v$. To achieve this, consider for each $(i,u)\in\dom(\amalg\bar t)$ the set $\Pre(i,u)$ of all $(i,u')\in\dom(\amalg\bar t)$ where $u'$ is a prefix of $u$. For $(i,u),(j,v)\in\dom(\amalg\bar t)$ we have $u=v$ iff $|u|=|v|$ and for all $(i,u')\in\Pre(i,u)$ and $(j,v')\in\Pre(j,v)$ with $|u'|=|v'|>0$ the last symbols of $u'$ and $v'$ coincide. Since the set $\Pre\bigl((i,u),X\bigr)$ is definable in $\MSO{}$, we can express this characterisation in $\MSO{}$ as well. Heavily using the resulting formula $E$, one can construct an interpretation $\langle f_\amalg,\Int_\amalg\rangle$ of $\otimes\bar t$ in $\amalg\bar t$ such that $f_\amalg(u)=(i,u)$, where $i$ is minimal with $u\in\dom(t_i)$.

\paragraph{The Interpretations $\langle f_{\otimes,i},\Int_{\otimes,i}\rangle$.} For all $\run i1n$ and $\bar w\in(\widehat\Sigma^\star)^n$ one can easily give an interpretation $\langle f_{\otimes,i},\Int_{\otimes,i}\rangle$ of $w_i$ in $H(\otimes\bar w)$ such that $f_{\otimes,i}(p)=(p-1)\cdot n+i$.

\paragraph{The Interpretation $\langle f_H,\Int_H\rangle$.} For $\run i1n$ let $\langle f_{C,i},\Int_C\rangle$ be the interpretation of~$t_i$ in~$C(t_i)$ from Lemma~\ref{lemma:interpretation_C}. Since the $f_{\otimes,i}$'s have mutually disjoint images, the map $f_H\colon\dom(\amalg\bar t)\to\dom\bigl(H(\otimes C(\bar t))\bigr)$ with $f_H(i,u)=f_{\otimes,i}(f_{C,i}(u))$ is one-to-one. For $(i,u)\in\dom(\amalg\bar t)$ we get $|u|\cdot K< f_{C,i}(u)\leq\bigl(|u|+1\bigr)\cdot K$ and hence
\begin{equation*}
	|u|\cdot K\cdot n < f_H(i,u) \leq \bigl(|u|+1\bigr)\cdot K\cdot n\,.
\end{equation*}
Exploiting this observation for the formula $L^{\Int_H}$ and using $\Int_C$ and $\Int_{\otimes,i}$, one can construct formulae $\Int_H$ such that $\langle f_H,\Int_H\rangle$ is an interpretation of $\amalg\bar t$ in $H\bigl(\otimes C(\bar t)\bigr)$.

\begin{proof}[of Proposition~\ref{prop:C_preserves_automaticity}]
Let $\Gamma_H$ be an $\MSO{\signW{\widehat\Sigma_\Box^n}}$-sentence defining the language $H\bigl(\otimes(\widehat\Sigma^\star)^n\bigr)\subseteq(\widehat\Sigma_\Box^n)^\star$. If $\Phi$ defines $\otimes R$, then
\begin{equation*}
	\Gamma_H \land \bigwedge\nolimits_{\run i1n} \Gamma_C^{\Int_{\otimes,i}} \land (\Phi^{\Int_\amalg})^{\Int_H}
\end{equation*}
defines $H\bigl(\otimes C(R)\bigr)$. Since $H$ is a one-to-one monoid morphism, $\otimes C(R)$ is regular as well. Finally, all employed constructions are effective.\qed
\end{proof}

\section{Fat Tree Automatic Ordinals Are Not Word Automatic}
\label{sec:fat_TA_ord_not_WA}

The goal of this section is to give the last missing piece for the proof of Theorem~\ref{thm:main}, namely the following theorem:

\begin{theorem}
\label{thm:fat_TA_ord_are_not_WA}
Let $\struct L$ be a tree automatic scattered linear ordering such that $\dom(\struct L)$ is fat. Then, $\struct L$ is \emph{not} word automatic.
\end{theorem}

\noindent The theorem below states the necessary condition on word automatic linear orderings we use to show non-automaticity:

\begin{theorem}[Khoussainov, Rubin, Stephan \cite{KRS03}]
\label{thm:WA_ord_FC_is_finite}
If $\struct L$ is a word automatic linear ordering, then its \textsl{FC}-rank is finite.
\end{theorem}

\noindent Actually, we do not need any details on the \textsl{FC}-rank (finite condensation rank) besides the fact that every scattered linear ordering $\struct L$, having the property that for each $r\geq 1$ at least one linear ordering from
\begin{equation*}
	\mathcal{N}_r = \Set{ \struct A_1\cdot\struct A_2\dotsm\struct A_r |
		\struct A_1,\dotsc,\struct A_r\in\bigl\{(\Nat;<),(\Nat;>)\bigr\} }
\end{equation*}
can be embedded into $\struct L$, has infinite \textsl{FC}-rank. The main idea of the proof is as follows:

\begin{lemma}
\label{lemma:fat_TA_ord_have_high_rank}
Let $\struct L=(L;<)$ be a tree automatic scattered linear ordering, $(\aut A;\aut A_<)$ an automatic presentation of $\struct L$, $n$ the number of states of $\aut A$, and $r\geq 1$. \linebreak[1] If there exists some tree $t\in L$ with $\diameter(t)\geq r\cdot 2^n$, then there are infinite linear orderings $\struct A_1,\dotsc,\struct A_r$ such that $\struct A_1\cdot\struct A_2\dotsm\struct A_r$ can be embedded into $\struct L$.
\end{lemma}

\newcommand{\state}[1]{\llbracket #1\rrbracket_<}
\newcommand{\bigstate}[1]{\bigl\llbracket #1\bigr\rrbracket_<}

\noindent For any linear ordering $\struct A$ and all $a_1,a_2\in\dom(\struct A)$ we define $\cmp_{\struct A}(a_1,a_2)\in\{-1,0,1\}$ to be $-1$ if $a_1 <^{\struct A} a_2$, $0$ if $a_1=a_2$, and $1$ if $a_2<^{\struct A}a_1$. To simplify notation, we put $\state{s,t}=\aut A_<\bigl(\otimes(s,t)\bigr)$ for all $s,t\in\trees\Sigma$. Moreover, we assume w.l.o.g. that from $\state{s,t}$ one can deduce whether $s=t$ holds true. Then, $\cmp_{\struct L}(s,t)$ \emph{is determined by} $\state{s,t}$ for all $s,t\in L$, i.e., there is a map $f$ from the state set of~$\aut A_<$ to $\{-1,0,1\}$ such that $\cmp_{\struct L}(s,t)=f\bigl(\state{s,t}\bigr)$ for all $s,t\in L$.

\begin{proof}
Let $\struct T\in L$ be a tree and $\ell\geq n$ such that $\bigl|\dom(\struct T)\cap\{0,1\}^\ell\bigr|\geq r\cdot 2^n$. Thus, there exist at least $r$ mutually distinct $u\in\dom(\struct T)\cap\{0,1\}^{\ell-n}$ for which there is a $v\in\{0,1\}^n$ with $uv\in\dom(\struct T)$, say $u_1,\dotsc,u_r$. For $\bar t=(t_1,\dotsc,t_r)\in(\trees\Sigma)^r$ let $\struct T[\bar t]\in\trees\Sigma$ be the tree obtained from $\struct T$ by replacing for each $\run i1r$ the subtree rooted at $u_i$ with $t_i$. Then, $\aut A\bigl(\struct T[\bar t]\bigr)$ is determined by the $r$ states $\aut A(t_1),\dotsc,\aut A(t_r)$ for all $\bar t\in(\trees\Sigma)^r$. Moreover, for $\bar s\in(\trees\Sigma)^r$ the tree $\otimes\bigl(\struct T[\bar s],\struct T[\bar t]\bigr)$ is obtained from $\otimes(\struct T,\struct T)$ by replacing for each $\run i1r$ the subtree rooted at $u_i$ with $\otimes(s_i,t_i)$. Consequently, $\bigstate{\struct T[\bar s],\struct T[\bar t]}$ is determined by the $r$ states $\state{s_1,t_1},\dotsc,\state{s_r,t_r}$ for all $\bar s,\bar t\in(\trees\Sigma)^r$.

Observe that $h(\struct T\restrict u_i)\geq n$ for each $\run i1r$. Therefore, by Lemma~\ref{lemma:inflang_characterization} and Ramsey's theorem for infinite, undirected, finitely coloured graphs, there exists an infinite set $A_i\subseteq\trees\Sigma$ of trees $t\in\trees\Sigma$ with $\aut A(t)=\aut A(\struct T\restrict u_i)$ such that
\begin{equation*}
	c(s,t) = \bigl\{\state{s,s},\state{t,t},\state{s,t},\state{t,s}\bigr\}
\end{equation*}
is the same set $Q_i$ for all distinct $s,t\in A_i$. It turns out that $Q_i$ has exactly three elements and $\state{s,s}=\state{t,t}$ for all $s,t\in A_i$.

Now, put $A=A_1\times\dotsb\times A_r$. For each $\bar t\in A$ we have $\aut A\bigl(\struct T[\bar t]\bigr)=\aut A(\struct T)$ and hence $\struct T[\bar t]\in L$. We define a linear ordering $\struct A=\bigl(A;<^{\struct A}\bigr)$ by $\bar s <^{\struct A} \bar t$ iff $\struct T[\bar s] < \struct T[\bar t]$. By definition, $\struct A$ can be embedded into $\struct L$.

For $\run i1r$, $\bar a\in A$, and $t\in A_i$ we let $\bar a_{i/t}\in A$ be the tuple $\bar a$ with the $i$-th component replaced by $t$. Then, for all $\bar a,\bar b$ and $s,t\in A_i$ we obtain $\bigstate{\struct T[\bar a_{i/s}],\struct T[\bar a_{i/t}]}=\bigstate{\struct T[\bar b_{i/s}],\struct T[\bar b_{i/t}]}$ and hence $a_{i/s} <^{\struct A} a_{i/t}$ iff $b_{i/s} <^{\struct A} b_{i/t}$. Thus, defining a linear ordering $\struct A_i=\bigl(A_i;<^{\struct A_i}\bigr)$ by $s <^{\struct A_i} t$ iff $\bar a_{i/s} <^{\struct A} \bar a_{i/t}$ is independent from the specific choice of $\bar a\in A$. Clearly, $\cmp_{\struct A_i}(s,t)$ is determined by $\state{s,t}$ for all $s,t\in A_i$. Since $Q_i$ contains exactly three elements, $\state{s,t}$ is determined by $\cmp_{\struct A_i}(s,t)$ for all $s,t\in A_i$ as well. Hence, the linear orderings $\struct A$ and $\struct A_1,\dotsc,\struct A_r$ satisfy the condition of Lemma~\ref{lemma:ordering} below and consequently $\struct A_{\pi(1)}\dotsm\struct A_{\pi(r)}$ can be embedded into $\struct L$.\qed
\end{proof}

\begin{lemma}
\label{lemma:ordering}
Let $\struct A$ and $\struct A_1,\dotsc,\struct A_r$ be infinite linear orderings with $\dom(\struct A)=\dom(\struct A_1)\times\dotsb\times\dom(\struct A_r)$ and satisfying the following two conditions:
\killenumspace
\begin{enumerate}
\item[(1)] $\cmp_{\struct A}(\bar a,\bar b)$ is determined by $\cmp_{\struct A_1}(a_1,b_1),\dotsc,\cmp_{\struct A_r}(a_r,b_r)$ for all $\bar a,\bar b\in A$,
\item[(2)] if $\bar a,\bar b\in A$ differ only in the $i$-th component, then $\cmp_{\struct A}(\bar a,\bar b)=\cmp_{\struct A_i}(a_i,b_i)$.
\end{enumerate}
\killenumspace
Then, there exists a permutation $\pi$ of $\{1,\dotsc,r\}$ such that $\struct A$ is isomorphic to $\struct A_{\pi(1)}\cdot\struct A_{\pi(2)}\dotsm\struct A_{\pi(r)}$.
\end{lemma}

\noindent Finally, we are in a position to prove Theorem~\ref{thm:fat_TA_ord_are_not_WA}.

\begin{proof}[of Theorem~\ref{thm:fat_TA_ord_are_not_WA}]
Let $(\aut A;\aut A_<)$ be an automatic presentation of $\struct L$ and $n$ the number of states of $\aut A$. Since $\dom(\struct L)$ is fat, for any $r\geq 1$ there is a $t\in\dom(\struct L)$ with $\diameter(t)\geq r\cdot 2^n$. Let $\struct A_1,\dotsc,\struct A_r$ be the infinite linear orderings from Lemma~\ref{lemma:fat_TA_ord_have_high_rank}. For each $i\in[1,r]$ some ${\struct B_i\in\bigl\{ (\Nat;<),(\Nat;>) \bigr\}}$ can be embedded into $\struct A_i$. Then, $\struct B_1\cdot\struct B_2\dotsm\struct B_r\in\mathcal{N}_r$ can be embedded into $\struct A_1\cdot\struct A_2\dotsm\struct A_r$ and consequently into $\struct L$. Hence, $\struct L$ has infinite \textsl{FC}-rank and is, by Theorem~\ref{thm:WA_ord_FC_is_finite}, \emph{not} word automatic.\qed

\end{proof}

\section{Conclusions}

Altogether, we proved that is decidable whether a given tree automatic scattered linear ordering is already word automatic. Taking a closer look at the proof reveals that the problem is solvable nondeterministically in logarithmic space, provided the tree automaton recognising the domain is reduced.

The restriction to scattered linear orderings naturally rises the question whether this result holds true for general linear orderings. Unfortunately, this problem cannot be solved by means of our technique since the ordering $(\mathbb{Q};<)$ of the rationals admits a word automatic as well as a fat tree automatic presentation. As the Boolean algebra of finite and co-finite subsets of $\Nat$ shares this feature, the same pertains to an analogue of Theorem~\ref{thm:main} for Boolean algebras. In spite of that, we suggest trying to apply the technique to other classes of structures, such as groups, for which a necessary condition on its automatic members is known.

Finally, Theorem~\ref{thm:main} provides a decidable characterisation of all tree automatic ordinals $\alpha\geq\omega^\omega$. Finding such a characterisation for each $\omega^{\omega^k}$ with $k\in\Nat$ possibly turns out to be the main ingredient for showing that the isomorphism problem for tree automatic ordinals is decidable.

\nocite{Blu99}

\bibliography{watasloid}
\bibliographystyle{splncs_srt}

\clearpage

\begin{appendix}

\section{Proofs of Lemmas~\ref{lemma:inflang_characterization} and~\ref{lemma:fatlang_characterization}}

\noindent Recall that we fixed a reduced tree automaton $\aut A=(Q,\iota,\delta,F)$ and defined the graph $G_{\aut A}=(Q,E_{\aut A})$ by
\begin{equation*}
	\tag{\ref{eq:def_E_A}}
	(p,q)\in E_{\aut A} \quad\text{iff}\quad \exists a\in\Sigma,r\in Q\colon
			\delta(a,p,r)=q \text{ or } \delta(a,r,p)=q\,.
\end{equation*}
An edge $(p,q)\in E_{\aut A}$ was called \emph{special} if in Eq.~\eqref{eq:def_E_A} the state $r\in Q$ can be chosen such that it satisfies the conditions of Lemma~\ref{lemma:inflang_characterization}. Moreover, we denote the \emph{prefix relation} on $\{0,1\}^\star$ by $\prefixeq$, i.e., $u\prefixeq v$ if there is some $w\in\{0,1\}^\star$ such that $uw=v$.

To simplify notation, we put $t[u]=\aut A(t\restrict u)$ for each $t\in\trees\Sigma$ and $u\in\dom(t)$. In particular, $t[\epsilon]=\aut A(t)$. For all $u\in\dom(t)$ with $u0,u1\in\dom(t)$ we have $\delta\bigl(t(u),t[u0],t[u1]\bigr)=t[u]$ and hence $\bigl(t[ud],t[u]\bigr)\in E_{\aut A}$ for both $d=0$ and $d=1$. We denote these edges by $t[ud,u]$. Applying this argument repeatedly, for all $u,v\in\dom(t)$ with $u\prefixeq v$ we obtain a path from $t[v]$ to $t[u]$ of length\footnote{The length of a path is the number of its edges.} $|v|-|u|$ in $G_{\aut A}$, which we denote by $t[v,u]$.

Conversely, let $(p,q)\in E_{\aut A}$ be an edge in $G_{\aut A}$ and $t\in\trees\Sigma$ a tree with $\aut A(t)=p$. Then, there are $a\in\Sigma$ and $r\in Q$ such that, w.l.o.g., $\delta(a,p,r)=q$. Moreover, there is a tree $s\in\trees\Sigma$ with $\aut A(s)=r$. Then, the unique tree $t'\in\trees\Sigma$ with $t'(\epsilon)=a$, $t'\restrict 0=t$, and $t'\restrict 1=s$ satisfies $\aut A(t')=q$ and $t'\restrict 0=t$. Applying this argument repeatedly, we obtain for each path from $p$ to $q$ of length $m$ in $G_{\aut A}$ and any tree $t\in\trees\Sigma$ with $\aut A(t)=p$ another tree $t'\in\trees\Sigma$ and a position $u\in\dom(t')$ such that $\aut A(t')=q$, $|u|=m$, and $t'\restrict u=t$.

\setcounter{mysection}{3}
\setcounter{lemma*}{2}
\begin{lemma*}
For every $q\in Q$ the following are equivalent:
\begin{enumerate}[(1)]
\item there are infinitely many $t\in\trees\Sigma$ satisfying $\aut A(t)=q$,
\item there is a tree $t\in\trees\Sigma$ satisfying $h(t)\geq n$ and $\aut A(t)=q$, where $n=|Q|$,
\item $G_{\aut A}$ contains a cycle from which $q$ is reachable.
\end{enumerate}
\end{lemma*}

\begin{proof}
Trivially, (1)~implies~(2). It remains to show that (2)~implies~(3) and (3)~implies~(1).

\paragraph{To (2) implies (3).} Let $t\in\trees\Sigma$ be a tree with $h(t)\geq n$ and $\aut A(t)=q$. Consider some $u\in\dom(t)$ such that $|u|=n$. Then, $t[u,\epsilon]$ is a path of length $n$ ending in~$q$. Due to the pigeonhole principle, this path contains a cycle.

\paragraph{To (3) implies (1).} It suffices to show that for each $m\geq 0$ there is a tree $t\in\trees\Sigma$ with $\aut A(t)=q$ and $h(t)\geq m$. Thus, consider some $m\geq 0$. There is a path of length $m$ ending in $q$. Let $p\in Q$ be the first state of this path and $s\in\trees\Sigma$ a tree with $\aut A(s)=p$. Then, there are a tree $t\in\trees\Sigma$ and $u\in\dom(t)$ with $\aut A(t)=q$, $|u|=m$, and $t\restrict u=s$. In particular, $h(t)\geq m$.\qed
\end{proof}

\pagebreak

\begin{lemma*}
The following are equivalent:
\begin{enumerate}[(1)]
\item the tree language $L$ recognised by $\aut A$ is fat,
\item there is a tree $t\in L$ satisfying $\diameter(t)> 2^{n-1}$, where $n=|Q|$,
\item $G_{\aut A}$ contains a cycle including a special edge and from which $F$ is reachable.
\end{enumerate}
\end{lemma*}

\begin{proof}
Trivially, (1)~implies~(2). It remains to show that (2)~implies~(3) and (3)~implies~(1).

\paragraph{To (2) implies (3).} Using induction on $m\geq 0$ we show the following: For every tree $t\in\trees\Sigma$ with $\diameter(t)>2^{m-1}$ and $|Q_t|\leq m$, where
\begin{equation*}
	Q_t = \bigl\{\; t[u] \bigm| u\in\dom(r) \;\bigr\}\,,
\end{equation*}
there are $u_1,u_2\in\dom(t)$ such that $u_1\prefix u_2$ and $t[u_2,u_1]$ is a cycle containing a special edge.

For $m=0$ there is nothing to show. Thus, let $m>0$. Consider an $\ell\geq0$ such that $|U|>2^{m-1}$ for $U=\dom(t)\cap\{0,1\}^\ell$. Let $u\in\dom(t)$ be the longest common prefix of all positions in $U$. Clearly, $\ell\geq |u|+m$. There are two cases:
\begin{enumerate}
\item There is a $v\in\dom(t)$ with $u\prefix v$ and $t[u]=t[v]$. W.l.o.g., we assume $u0\prefixeq v$. By the choice of $u$, there is some $w\in U$ with $u1\prefixeq w$. The path $t[w,u]$ has length $\ell-|u|\geq m$ and hence contains a cycle. Since $t[u1]$ lies on or after this cycle, i.e., $t[u1]$ satisfies condition~(3) of Lemma~\ref{lemma:inflang_characterization}. Thus, the edge $t[u0,u]\in E_{\aut A}$, which is contained in the cycle $t[v,u]$, is special.
\item There is no $v\in\dom(t)$ with $u\prefix v$ and $t[u]=t[v]$. In particular, $2\leq|Q_t|\leq m$. Since $\diameter(t\restrict u)\geq |U|>2^{m-1}$, we have $\diameter(t\restrict u0)>2^{m-2}$ or $\diameter(t\restrict u1)>2^{m-2}$. W.l.o.g., assume $\diameter(s)>2^{m-2}$ for $s=t\restrict u0$. We have $t[u]\not\in Q_s$ and hence $|Q_s|<|Q_t|\leq m$. By the induction hypothesis, there are $v,w\in\dom(s)$ such that $v\prefix w$ and $s[w,v]$ is a cycle containing a special edge. The claim follows from $u0v\prefix u0w$ and $t[u0w,u0v]=s[w,v]$.
\end{enumerate}

\noindent Finally, consider some $t\in L(\aut A)$ with $\diameter(t)>2^{n-1}$. Obviously, $|Q_t|\leq|Q|=n$. From the cycle $t[u_2,u_1]$ we can reach a state from $F$ along the path $t[u_1,\epsilon]$.

\paragraph{To (3) implies (1).} Using induction on $m\geq 0$ we show that if there is a path containing $m$ special edges and which ends in some $q\in Q$, then there is a tree $t\in\trees\Sigma$ with $\aut A(t)=q$ and $\diameter(t)>m$. Due to the cycle there are paths containing arbitrarily many special edges and which end in $F$. Thus, condition~(1) will follow.

For $m=0$ any tree $t\in\trees\Sigma$ with $\aut A(t)=q$ trivially satisfies $\diameter(t)>0$. Thus, consider $m>0$. Let $(p,r)\in E_{\aut A}$ be the last special edge in the path. By the induction hypothesis, there is a tree $s\in\trees\Sigma$ with $\aut A(s)=p$ and $\diameter(s)> m-1$. Let $\ell\geq 0$ be such that $|\dom(s)\cap\{0,1\}^\ell|=\diameter(s)$. Moreover, there are $a\in\Sigma$ and $p'\in Q$ such that, w.l.o.g., $\delta(a,p,p')=r$ and $p'$ satisfies the conditions of Lemma~\ref{lemma:inflang_characterization}. Thus, there is a tree $s'\in\trees\Sigma$ such that $h(s')\geq\ell$. Then, the unique tree $t'\in\trees\Sigma$ with $t'(\epsilon)=a$, $t'\restrict 0=s$, and $t'\restrict 1=s'$ satisfies $\aut A(t')=r$. Since there is a path from $r$ to $q$, say it has $n$ edges, there are $t\in\trees\Sigma$ and $u\in\dom(t)$ such that $\aut A(t)=q$, $|u|=n$, and $t\restrict u=t'$. From the construction of $t$ we obtain
\begin{multline*}
	\dom(t)\cap\{0,1\}^{n+1+\ell} \supseteq u\bigl(\dom(t')\cap\{0,1\}^{1+\ell}\bigr) \\
		= u0\bigl(\dom(s)\cap\{0,1\}^\ell\bigr)\cup u1\bigl(\dom(s')\cap\{0,1\}^\ell\bigr)\,.
\end{multline*}
Since the union on the right hand side is disjoint and the set $\dom(s')\cap\{0,1\}^\ell$ is not empty, we have
\begin{equation*}
	\bigl|\dom(t)\cap\{0,1\}^{n+1+\ell}\bigr| \geq
		\bigl|\dom(s)\cap\{0,1\}^\ell\bigr| + \bigl|(\dom(s')\cap\{0,1\}^\ell\bigr| > m\,,
\end{equation*}
i.e., $\diameter(t)>m$.\qed
\end{proof}

\clearpage

\section{Interpretations and Formulae from Section \ref{sec:slim_TA_are_WA}}

\subsection{The Interpretation $\langle f_C,\Int_C\rangle$}

\setcounter{mysection}{4}
\setcounter{lemma*}{2}

\begin{lemma*}
For all $t\in\trees[K]\Sigma$ there is an $\MSO{}$-interpretation $\langle f_C,\Int_C\rangle$ of $t$ in~$C(t)$ such that $\Int_C$ does not depend on $t$.
\end{lemma*}

\noindent The formulae $\Int_C=\bigl(\Delta_C;(S_d^{\Int_C})_{d\in\{0,1\}},(P_a^{\Int_C})_{a\in\Sigma}\bigr)$ are as follows:
\begin{align*}
	\Delta_C(x) &= \neg P_\$(x) \\
	S_d^{\Int_C}(x,y) &= \bigvee\nolimits_{a\in\Sigma} P_{(a,1)}(x)\land\Delta_C(y)\land\exists z\Bigl(\phi(x,z)\land
		\bigvee\nolimits_{1\leq s\leq\frac{K}{2}} 	\psi_{d,s}(x,y,z)\Bigr) \\
	P_a^{\Int_C}(x) &= \Delta_C(x)\land\bigl(P_{(a,0)}(x)\lor P_{(a,1)}(x)\bigr) \\
	\phi(x,z) &= z\equiv 1\,(\bmod\,K)\land z\leq x<z+K \\
	\psi_{d,s}(x,y,z) &= \exists^{=s} z'\Bigl(z\leq z'\leq x\land\bigvee\nolimits_{a\in\Sigma} P_{(a,1)}(z')\Bigr)\land
		y = z+K+2s-2+d
	\tag{$1\leq s\leq\frac{K}{2}$}
\end{align*}

\subsection{Proof of Lemma~\ref{lemma:characterisation_image_C}}
\label{subsec:characterisation_image_C}

\setcounter{mysection}{4}
\setcounter{lemma*}{4}

\begin{lemma*}
Let $\sigma\in\widehat\Sigma^\star$. There exists a tree $t\in\trees[K]\Sigma$ with $C(t)=\sigma$ iff $\sigma=\sigma_0\sigma_1\dotso\sigma_n$ for some $n\geq 0$ and $\sigma_0,\dotsc,\sigma_n\in\widehat\Sigma^K$ satisfying (a) and (b):
\begin{enumerate}[(a)]
\item $\sigma_\ell=\alpha_{\ell,1}\dotso\alpha_{\ell,s_\ell}\$^{K-s_\ell}$ for some $s_\ell\geq 1$ and $\alpha_{\ell,1},\dotsc,\alpha_{\ell,s_\ell}\in\Sigma\times\{0,1\}$ and for each $\run \ell0n$,
\item $s_0=1$, $s_{\ell+1}=2\cdot(c_{\ell,1}+\dotsb+c_{\ell,s_\ell})$ for $0\leq\ell<n$, and $c_{m,1}+\dotsb+c_{m,s_m}=0$, where $\alpha_{\ell,r}=\langle a_{\ell,r},c_{\ell,r}\rangle$.
\end{enumerate}
\end{lemma*}

\begin{proof}
To see that $C(t)$ has the required shape, choose $n=m$, $\sigma_\ell$ and $s_\ell$ as in the construction of $C(t)$, and $\alpha_{\ell,r}=\langle t(u_{\ell,r}),c_{\ell,r}\rangle$. Then, condition (a) is trivially met, whereas (b) is satisfied since each tree has exactly one node on the zeroth level, on each other level twice as many nodes as inner nodes on the previous level, and no inner nodes on the last level.

Conversely, consider some $\sigma=\alpha_1\dotso\alpha_{(n+1)\cdot K}\in\widehat\Sigma^\star$ of the required shape. Let $\struct T=\Int_C(\sigma)$ and $T=\dom(\struct T)$. First, we observe that
\begin{equation*}
	T = \Set{ \ell\cdot K+r | \run\ell0n,\run r1{s_\ell} }
\end{equation*}
and put $(a_p,c_p)=\alpha_p$ for all $p\in T$. For $p=\ell\cdot K+r,q\in T$ and $d\in\{0,1\}$ with $\run\ell0n$ and $\run r1{s_\ell}$ we have $(p,q)\in S_d^{\struct T}$ iff $c_p=1$ and $q=(\ell+1)\cdot K+2s-1+d$, where $s=c_{\ell \cdot K+1}+\dotsb+c_p$. Notice that this reflects the introductory observation in the proof of Lemma~\ref{lemma:interpretation_C}.

Second, we construct a map $f\colon T\to\{0,1\}^\star$. We define the value $f(q)$ by induction on $q\in T$. The resulting map will satisfy $|f(q)|=\ell$ for all $q\in T$ with $\ell\cdot K< q\leq (\ell+1)\cdot K$. The first condition of (b) yields $q=1$ or $q>K$. We put $f(1)=\epsilon$. For $q>K$ there are unique $\ell,r\in\Nat$ with $0\leq\ell<n$ and $1\leq s\leq\frac{s_{\ell+1}}{2}$, and $d\in\{0,1\}$ such that $q=(\ell+1)\cdot K+2s-1+d$. From the second condition of (b) we conclude $s\leq c_{\ell\cdot K+1}+\dotsc+c_{\ell\cdot K+s_\ell}$. Thus, there is a least $\run p{\ell\cdot K+1}{\ell\cdot K+s_\ell}$ such that $s=c_{\ell\cdot K+1}+\dotsb+c_p$ and the minimality implies $c_p=1$. Since $p<q$, we are allowed to put $f(q)=f(p)d$.

A simple but tedious inspection of this construction shows for all $p,q\in T$ that $p<q$ iff $f(p)<_{\operatorname{llex}} f(q)$, where $<_{\operatorname{llex}}$ is the length-lexicographic order on~$\{0,1\}^\star$. In particular, $f$ is one-to-one. Obviously, the set $D=f(T)$ is non-empty and finite. Due to the construction of $f$ it is also prefix-closed. For $u0\in D$ we have $f^{-1}(u0)+1\in T$ and $u1=f\bigl(f^{-1}(u0)+1\bigr)\in D$. Similarly, $u1\in D$ implies $u0\in D$. Thus, $D$ is a tree domain. We define a tree $t\in\trees\Sigma$ by $\dom(t)=D$ and $t(u)=a_{f^{-1}(u)}$. It turns out that $f$ induces an isomorphism $f\colon\struct T\to t$ between $\signT\Sigma$-structures.

From the earlier remark on $|f(q)|$ we conclude $\diameter(t)\leq K$. Moreover, $h(t)=n$ and $f(\ell\cdot K+1),\dotsc,f(\ell\cdot K+s_\ell)$ is the lexicographic enumeration of $\dom(t)\cap\{0,1\}^\ell$ for all $0\leq\ell\leq n$. For $p\in T$ with $p\leq n\cdot K$ the specific choice of $f$ yields $c_p=1$ iff $f(p)\{0,1\}\subseteq\dom(t)$. For $p>n\cdot K$ this also holds true since the third condition of (b) implies $c_p=0$ and $|f(p)d|>n$ implies $f(p)d\not\in\dom(t)$ for $d\in\{0,1\}$. Altogether, we obtain $C(t)=\sigma$.\qed
\end{proof}

\subsection{The Sentence $\Gamma_C$}

The sentence $\Gamma_C$ is defined as follows:
\begin{align*}
	\Gamma_C &= \Phi_1\land\Phi_2\land\Phi_3\land\Phi_4\land\Phi_5 \\
	\Phi_1 &= \exists x\bigl(\forall y(y\leq x)\land x\equiv K\,(\bmod\, K)\bigr) \\
	\Phi_2 &= \forall x\bigl(x\equiv 1\,(\bmod\,K)\to\exists y\,\phi(x,y)\bigr) \\
	\Phi_3 &= \exists x\bigl(x=1\land\phi(x,x)\bigr) \\
	\Phi_4 &= \forall x_1\forall x_2\Bigl(x_1\equiv 1\,(\bmod\,K)\land x_2=x_1+K\to
		\bigvee\nolimits_{1\leq s\leq\frac{K}{2}} \psi_s(x_1,x_2)\Bigr) \\
	\Phi_5 &= \exists x\Bigl(x\equiv 1\,(\bmod\, K) \land \neg\exists z(z=x+K) \land \exists y\bigl(\phi(x,y)\land\chi_0(x,y)\bigr)\Bigr) \\
	\phi(x,y) &= x\leq y<x+K \land \forall z\bigl(x\leq z<x+K\to(P_\$(z)\leftrightarrow y<z)\bigr) \\
	\psi_s(x) &= \exists y_1\exists y_2\bigl(\phi(x_1,y_1)\land \phi(x_2,y_2)\land \chi_s(x_1,y_1)\land y_2=x_2+2s-1\bigr)
	\tag{$1\leq s\leq\frac{K}{2}$} \\
	\chi_s(x,y) &= \exists^{=s} z\Bigl(x\leq z\leq y \land \bigvee\nolimits_{a\in\Sigma} P_{(a,1)}(z)\Bigr)
	\tag{$0\leq s\leq\frac{K}{2}$}
\end{align*}
The purpose of the formulae $\Phi_i$ is as follows:
\begin{description}
\item[$\Phi_1$ --] the length $|\sigma|$ of $\sigma$ is a positive multiple of $K$, say $|\sigma|=(m+1)\cdot K$
\item[$\Phi_2$ --] $\sigma$ can be written as $\sigma=\sigma_0\dotso\sigma_m$ with $\sigma_0,\dotsc,\sigma_m\in\widehat\Sigma^K$ such that (a) is satisfied, therein $\phi(\ell\cdot K+1,p)$ holds precisely for $p=\ell\cdot K+s_\ell$
\item[$\Phi_3$ --] the first condition of (b) is satisfied
\item[$\Phi_4$ --] the second condition of (b) is satisfied
\item[$\Phi_5$ --] the third condition of (b) is satisfied
\end{description}

\subsection{The Formula $E$ and the Interpretation $\langle f_\amalg,\Int_\amalg\rangle$}

The formula $E(x_1,x_2)$ is defined as follows:
\newcommand{\Cl}{\operatorname{Cl}}
\begin{align*}
	E(x_1,x_2) &= L(x_1,x_2)\land\exists X_1\exists X_2\bigl(\phi_{\Pre}(x_1,X_1)\land\phi_{\Pre}(x_2,X_2)\land \psi(X_1,X_2)\bigr)\\
	\phi_{\Pre}(x,X) &= \phi_{\Cl}(x,X)\land\forall Y\bigl(\phi_{\Cl}(x,Y)\to \forall y(y\in X\to y\in Y)\bigr) \\
	\phi_{\Cl}(x,X) &= x\in X\land\forall y\forall z\bigl(S(y,z)\land z\in X\to y\in X\bigr) \\
	S(y,z) &= S_0(y,z)\lor S_1(y,z) \\
	\psi(X_1,X_2) &= \forall z_1\forall z_2\bigl(z_1\in X_1\land z_2\in X_2\land L(z_1,z_2)\land \bigl(\exists y\, S(y,z_1)\bigr)\to\chi(z_1,z_2)\bigr) \\
	\chi(z_1,z_2) &= \exists y_1\exists y_2\bigl(\bigl(S_0(y_1,z_1)\land S_0(y_2,z_2)\bigr)\lor\bigl(S_1(y_1,z_1)\land S_1(y_2,z_2)\bigr)\bigr)
\end{align*}
The ideas behind these formulae are the following:
\begin{description}
\item[$\phi_{\Cl}(x,X)$ --] the set $X$ contains $x$ and is closed under taking predecessors
\item[$\phi_{\Pre}(x,X)$ --] the set $X$ is the smallest one (w.r.t. inclusion) having the property $\phi_{\Cl}(x,X)$ and hence $\amalg\bar t\models\phi_{\Pre}\bigl((i,u),U\bigr)$ iff $U=\Pre(i,u)$
\item[$\psi(X_1,X_2)$ --] describes the condition on $\Pre(i,u)$ and $\Pre(j,v)$ in the characterisation of $|u|=|v|$ in terms of these two sets
\end{description}

\begin{lemma}
For all $\bar t\in(\trees\Sigma)^n$ there is an $\MSO{}$-interpretation $\langle f_\amalg,\Int_\amalg\rangle$ of~$\otimes\bar t$ in~$\amalg\bar t$ such that $f_\amalg(u)=(i,u)$, where $i$ is minimal with $u\in\dom(t_i)$, and $\Int_\amalg$ does not depend on $\bar t$.
\end{lemma}

\noindent The formulae $\Int_\amalg=\Bigl(\Delta_{\Int_\amalg};\bigl(S_d^{\Int_\amalg}\bigr)_{d\in\{0,1\}},\bigl({P_{\overline\alpha}}^{\Int_\amalg})_{\overline\alpha\in\widehat\Sigma_\Box^n}\Bigr)$ are as follows:
\begin{align*}
	\Delta_{\Int_\amalg}(x) &= \bigvee\nolimits_{1\leq i\leq n} \Bigl(Q_i(x)\land\bigwedge\nolimits_{1\leq j<i}\phi_{j,\Box}(x)\Bigr) \\
	S_d^{\Int_\amalg}(x,y) &= \Delta_{\Int_\amalg}(x)\land\Delta_{\Int_\amalg}(y)\land\exists z\bigl(S_d(z,y)\land E(x,z)\bigr) \\
	P_{(\alpha_1,\dotsc,\alpha_n)}^{\Int_\amalg}(x) &= \Delta_{\Int_\amalg}(x)\land\bigvee\nolimits_{1\leq i\leq n} \phi_{i,\alpha_i}(x) \\
	\phi_{i,a}(x) &= \exists y\bigl(Q_i(y)\land P_a(y)\land E(x,y)\bigr)
		\tag{$\run i1n,a\in\Sigma$} \\
	\phi_{i,\Box}(x) &= \neg\exists y\bigl(Q_i(y)\land E(x,y)\bigr)
		\tag{$\run i1n$}
\end{align*}

\subsection{The Interpretations $\langle f_{\otimes,i},\Int_{\otimes,i}\rangle $}

\begin{lemma}
\label{lemma:interpretation_otimes}
For each $\run i1n$ and all $\bar w\in(\widehat\Sigma^\star)^n$ there is an $\MSO{}$-interpre\-tation $\langle f_{\otimes,i},\Int_{\otimes,i}\rangle$ of $w_i$ in $H(\otimes\bar w)$ such that $f_{\otimes,i}(p)=(p-1)\cdot n+i$ and $\Int_{\otimes,i}$ does not depend on $\bar w$.
\end{lemma}

\noindent The formulae $\Int_{\otimes,i}=\Bigl(\Delta_{\Int_{\otimes,i}};\leq^{\Int_{\otimes,i}},\bigl(P_\alpha^{\Int_{\otimes,i}}\bigr)_{\alpha\in\widehat\Sigma}\Bigr)$ are as follows:
\begin{align*}
	\Delta_{\Int_{\otimes,i}}(x) &= x\equiv i\,(\bmod\, n) \land \neg\psi_{i,\Box}(x) \\
	\leq^{\Int_{\otimes,i}}(x,y) &= \Delta_{\Int_{\otimes,i}}(x)\land\Delta_{\Int_{\otimes,i}}(y)\land x\leq y \\
	P_\alpha^{\Int_{\otimes,i}}(x) &= \Delta_{\Int_{\otimes,i}}(x) \land \psi_{i,\alpha}(x) \\
	\psi_{i,\alpha}(x) &= \bigvee\nolimits_{\overline\alpha=(\alpha_1,\dotsc,\alpha_n)\in\widehat\Sigma_\Box^n,\alpha_i=\alpha} P_{\overline\alpha}(x)\tag{$\alpha\in\widehat\Sigma_\Box$}
\end{align*}

\subsection{The Interpretation $\langle f_H,\Int_H\rangle$}

\begin{lemma}
\label{lemma:interpretation_H}
For all $\bar t\in(\trees[K]\Sigma)^n$ there is an $\MSO{}$-interpretation $\langle f_H,\Int_H\rangle$ of~$\amalg\bar t$ in $H\bigl(\otimes C(\bar t)\bigr)$ such that $\Int_H$ does not depend on $\bar t$.
\end{lemma}

\noindent The formulae $\Int_H=\Bigl(\Delta_{\Int_H};\bigl(S_d^{\Int_H}\bigr)_{d\in\{0,1\}},\bigl(P_a^{\Int_H}\bigr)_{a\in\Sigma},L^{\Int_H},\bigl(Q_i^{\Int_H}\bigr)_{1\leq i\leq n}\Bigr)$ are as follows:
\begin{align*}
	\Delta_{\Int_H}(x) &= \bigvee\nolimits_{1\leq i\leq n} \Delta_{\Int_{\otimes,i}}(x) \\
	S_d^{\Int_H}(x,y) &= \bigvee\nolimits_{1\leq i\leq n} \bigl(\Delta_{\Int_{\otimes,i}}(x)\land\Delta_{\Int_{\otimes,i}}(y)\land (S_d^{\Int_C})^{\Int_{\otimes,i}}(x,y)\bigr) \\
	P_a^{\Int_H}(x) &= \bigvee\nolimits_{1\leq i\leq n} \bigl(\Delta_{\Int_{\otimes,i}}(x)\land (P_a^{\Int_C})^{\Int_{\otimes,i}}(x)\bigr) \\
	L^{\Int_H}(x,y) &= \Delta_{\Int_H}(x)\land\Delta_{\Int_H}(y)\land\exists z\bigl(z\equiv 1\,(\bmod\,{K\cdot n})\land z\leq x,y<z+K\cdot n\bigr) \\
	Q_i^{\Int_H}(x) &= \Delta_{\Int_{\otimes,i}}(x)
\end{align*}

\subsection{The Sentence $\Gamma_H$}

The sentence $\Gamma_H$ is defined as follows:
\begin{align*}
	\Gamma_H &= \Phi \land \bigwedge\nolimits_{0\leq i<n} \Psi_i \\
	\Phi &= \forall x\forall y\Bigl(x\equiv 1\,(\bmod\, n)\land x\leq y<x+n\to\bigvee\nolimits_{\alpha\in\widehat\Sigma_\Box^n}\bigl(P_\alpha(x)\land P_\alpha(y)\bigr)\Bigr) \\
	\Psi_i &= \forall x\forall y\bigl(x\leq y\land\psi_{i,\Box}(x)\to\psi_{i,\Box}(y)\bigr)
	\tag{$\run i1n$}
\end{align*}
where $\psi_{i,\Box}$ is the formula from the construction of $\Int_{\otimes,i}$ for $\run i1n$. The formula $\Phi$ defines the set $H\bigl((\widehat\Sigma_\Box^n)^\star\bigr)$, whereas $\Psi_i$ defines the set of all words $w\in(\widehat\Sigma_\Box^n)^\star$ with $\pi_i(w)\in\widehat\Sigma^\star\Box^\star$, where $\pi_i\colon (\widehat\Sigma_\Box^n)^\star\to(\widehat\Sigma_\Box)^\star$ is the projection to the $i$-th component. Thus, the conjunction $\bigwedge_{\run i1n} \Psi_i$ defines the language $\otimes(\widehat\Sigma^\star)^n$. Since $H\bigl(\otimes(\widehat\Sigma^\star)^n\bigr)=H\bigl((\widehat\Sigma_\Box^n)^\star\bigr)\cap\otimes(\widehat\Sigma^\star)^n$, $\Gamma_H$ defines this set.

\clearpage

\section{Proof of Lemma~\ref{lemma:ordering}}

\setcounter{mysection}{5}
\setcounter{lemma*}{3}

\begin{lemma*}
Let $\struct A$ and $\struct A_1,\dotsc,\struct A_r$ be infinite linear orderings with $\dom(\struct A)=\dom(\struct A_1)\times\dotsb\times\dom(\struct A_r)$ and satisfying the following two conditions:
\killenumspace
\begin{enumerate}
\item[(1)] $\cmp_{\struct A}(\bar a,\bar b)$ is determined by $\cmp_{\struct A_1}(a_1,b_1),\dotsc,\cmp_{\struct A_r}(a_r,b_r)$ for all $\bar a,\bar b\in A$,
\item[(2)] if $\bar a,\bar b\in A$ differ only in the $i$-th component, then $\cmp_{\struct A}(\bar a,\bar b)=\cmp_{\struct A_i}(a_i,b_i)$.
\end{enumerate}
\killenumspace
Then, there exists a permutation $\pi$ of $\{1,\dotsc,r\}$ such that $\struct A$ is isomorphic to $\struct A_{\pi(1)}\cdot\struct A_{\pi(2)}\dotsm\struct A_{\pi(r)}$.
\end{lemma*}

\begin{proof}
Let $\struct A=(A;<)$ and $\struct A_i=(A_i;<_i)$ for each $\run i1r$. Consider $\bar a,\bar b\in A$ with $a_i \leq_i b_i$ for all $\run i1r$. If we put $\bar c^{(k)}=(b_1,\dotsc,b_k,a_{k+1},\dotsc,a_r)$ and apply condition (2) repeatedly, we obtain
\begin{equation*}
	\label{eq:star}
	\tag{$\star$}
	\bar a = \bar c^{(0)} \leq \bar c^{(1)} \leq \dotsb \leq \bar c^{(r)} = \bar b\,.
\end{equation*}

\noindent For the rest of this proof, fix $\bar x,\bar y,\bar z\in A$ such that $x_i <_i y_i <_i z_i$ for all $1\leq i\leq r$ and put $\bar e^{(i)}=\bar x_{i/y_i}$. Since the $\bar e^{(i)}$ are mutually distinct, there is a permutation $\pi$ of $\{1,\dots,r\}$ such that
\begin{equation*}
	\bar e^{(\pi(1))} > \bar e^{(\pi(2))} > \dotsb > \bar e^{(\pi(r))}\,.
\end{equation*}
Due to condition (1), this permutation does not depend on the specific choice of $\bar x$ and $\bar y$ as long as $x_i <_i y_i$ for all $\run i1r$. To see that this permutation has the desired property, we need show that the one-to-one correspondence
\begin{equation*}
	f\colon\struct A\to\struct A_{\pi(1)}\dotsm\struct A_{\pi(r)},
	(a_1,\dotsc,a_r)\mapsto (a_{\pi(1)},\dotsc,a_{\pi(r)})
\end{equation*}
is an isomorphism. However, to keep notation clear, we simply assume that $\pi$ is the identity.

Therefore, it suffices to show that $\bar a <^{\struct A_1\dotsm\struct A_r} \bar b$ implies $\bar a < \bar b$ for all $\bar a,\bar b\in A$. Thus, consider $\bar a,\bar b\in A$ satisfying the premise. Due to the definition of $\struct A_1\dotsm \struct A_r$, there exists a $\run k1r$ such that $a_k <_k b_k$ and $a_i=b_i$ for all $\run i1{k-1}$. For all $\run i1r$ we let $g_i=\min_{\struct A_i}(a_i,b_i)$, $h_i=\max_{\struct A_i}(a_i,b_i)$
\begin{gather*}
	\bar c^{(i)}=(g_1,\dotsc,g_i,h_{i+1},\dotsc,h_r)\in A\,,
\intertext{and}
	\bar d^{(i)}=(g_1,\dotsc,g_{i-1},h_i,g_{i+1},\dotsc,g_r)\in A\,.
\end{gather*}
Using induction, we show $c^{(i)} < d^{(i)}$ for all $\run i1r$ with $g_i\not=h_i$. Clearly, $g_i\not=h_i$ implies $g_i <_i h_i$. For the largest $i$ with this property, $c^{(i)} < d^{(i)}$ directly follows from condition (2). For all other $i$ with $g_i\not=h_i$ consider the least $j>i$ with $g_j\not=h_j$. We obtain the following chain of inequalities
\begin{equation*}
	\bar c^{(i)}_{i/x_i,j/y_j} = \bar c^{(j)}_{i/x_i,j/y_j} <
	\bar c^{(j)}_{i/y_i,j/x_j} < \bar d^{(j)}_{i/y_i,j/y_j} <
	\bar d^{(i)}_{i/z_i,j/x_j}\,,
\end{equation*}
where the first and last inequality are due to $\bar e^{(j)}<\bar e^{(i)}$ and condition~(1), whereas the second one is implied by the induction hypothesis $\bar c^{(j)}<\bar d^{(j)}$ and~(1). Finally, a last application of (1) yields $\bar c^{(i)} < \bar d^{(i)}$. Altogether, we obtain
\begin{equation*}
	\bar a \leq \bar c^{(k)} < \bar d^{(k)} \leq \bar b\,,
\end{equation*}
where both non-strict inequalities use \eqref{eq:star}.\qed
\end{proof}

\end{appendix}

\end{document}